\documentclass[12pt]{article}
\usepackage{paper}

\newcommand*{\supfrac}[2]{\frac{#1}{#2}}  
\newcommand*{\subfrac}[2]{\frac{#1}{#2}}  
\newcommand*{\isupfrac}[2]{#1/#2}  
\newcommand*{\isubfrac}[2]{#1/#2}  

\title{On the mass dependence of the modular operator for a double cone}
\author{%
	Henning Bostelmann\footnote{University of York, Department of Mathematics, York YO10 5DD, United Kingdom; e-mail: \url{henning.bostelmann@york.ac.uk}}, 
	Daniela Cadamuro\footnote{Institut f\"ur theoretische Physik, Universit\"at Leipzig, Br\"uderstra\ss e 16, 04103 Leipzig, Germany; e-mail: \url{cadamuro@itp.uni-leipzig.de}}, 
	Christoph Minz\footnote{Institut f\"ur theoretische Physik, Universit\"at Leipzig, Br\"uderstra\ss e 16, 04103 Leipzig, Germany; e-mail: \url{christoph.minz@itp.uni-leipzig.de}}}
\date{}

\usepackage{csquotes}  
\usepackage[%
backend=biber,
style=alphabetic,  
citestyle=alphabetic,  
sorting=nyt,
maxnames=10,  
giveninits=true,  
eprint=true,  
doi=false,  
url=false  
]{biblatex}
\graphicspath{{figures/}} 

\begin{document}

\maketitle

\begin{abstract}
	We present a numerical approximation scheme for the Tomita-Takesaki modular operator of local subalgebras in linear quantum fields, working at one-particle level. This is applied to the local subspaces for double cones in the vacuum sector of a massive scalar free field in $(1+1)$- and $(3+1)$-dimensional Minkowski spacetime, using a discretization of time-0 data in position space. In the case of a wedge region, one component of the modular generator is well known to be a mass-independent multiplication operator; our results strongly suggest that for the double cone, the corresponding component is still at least close to a multiplication operator, but that it is dependent on mass and angular momentum.
\end{abstract}

\section{Introduction}

Since its inception, Tomita-Takesaki modular theory \cite{1970Takesaki} has found important applications in the mathematical formulation of quantum physics. 
This applies both to quantum thermodynamics \cite{1967HaagHugenholtzWinnink,1979BratteliRobinson}, where the modular group is linked to time evolution, and to quantum field theory, where the modular group of local algebras associated with spacelike wedges can be identified with the symmetry group of boosts \cite{1975BisognanoWichmann}; 
this has lead to many more structural insights into quantum field theory \cite{2000Borchers}. 
More recently, interest has arisen because of the importance of the modular generator in relativistic quantum information theory (see, e.g., \cite{2018HollandsSanders,2019CasiniGrilloPontello,2020CeyhanFaulkner,2020CiolliLongoRuzzi}).

However, beyond the wedge case, the concrete form of the modular operator for local algebras $\mathcal{A}( \mathcal{O} )$ in quantum field theory often remains elusive. More concrete information is available in conformal theories \cite{1993BrunettiGuidoLongo,1993FroehlichGabbiani}, though.
The situation also simplifies in \emph{linear} quantum field theories in (pure) quasifree states, where the Tomita operator is of ``second quantized'' form; that is, finding the modular objects for a Weyl subalgebra $\mathcal{A}( \mathcal{O} )$ reduces to finding corresponding ``one-particle'' modular objects for a symplectic subspace $\mathcal{L}( \mathcal{O} )$.
For concreteness, consider a real scalar free field of mass $m \geq 0$ in $d$-dimensional Minkowski spacetime $\Minkowski^{d}$ in its time-0 formulation, and the region $\mathcal{O} \subset \Minkowski^{d}$ being the causal closure of a ``base'' region $\mathcal{B} \subset \Reals^{d - 1}$ in the time-0 plane; thus vectors in $\mathcal{L}( \mathcal{O} )$ are pairs $( f, g )$, with $f, g \in \mathcal{C}_0^\infty( \mathcal{B}, \Reals )$ the initial data of the wave equation.
In this context, the one-particle modular generator acts as a block matrix 
\begin{align}
\label{eq:ModularGenerator.BlockForm}
			- \i \log \Delta
	&= 
			\begin{pmatrix}
				0 & M_- \\
				-M_+ & 0
			\end{pmatrix}
\end{align}
with the operators 
\begin{align}
\label{eq:ModularGenerator.Blocks}
			M_\pm
	&= 2 A^{\pm \supfrac{1}{4}}
			\operatorname{arcoth}\left(
				A^{\supfrac{1}{4}} \chi A^{-\supfrac{1}{4}}
			+ A^{-\supfrac{1}{4}} \chi A^{\supfrac{1}{4}}
			- 1
			\right)
			A^{\pm \supfrac{1}{4}}
	,
\end{align}
where $A$ is the Helmholtz operator $-\Laplace + m^2$, and $\chi$ multiplies with the characteristic function of $\mathcal{B}$. (We will justify this formula in \autoref{prop:ModularFormula}, including some details on the domain of the factors $\chi$).

Equation~\eqref{eq:ModularGenerator.Blocks}, or similar formulas \cite{1989FiglioliniGuido,2020CiolliLongoRuzzi}, determines $\log \Delta$ algorithmically from $A$ and $\chi$. However, evaluating it in examples would require explicit knowledge of the spectral decomposition of the operator in the argument of $\operatorname{arcoth}(\cdot)$, and thus is usually not feasible in practice.

In special situations, as mentioned, $M_\pm$ can be described more explicitly: If $\mathcal{O}$ is a spacelike wedge in $x_1$-direction, then $M_-$ is multiplication with the function $M_-(\mathbf{x})=2\pi x_1$, independent of $m$.
In the massless case ($m=0$), if $\mathcal{B}$ is the ball of radius $r$ (thus $\mathcal{O}$ a double cone), one has $M_-(\mathbf{x}) = \pi \left( r^2 - \| \mathbf{x} \|^2 \right)$ \cite{1982HislopLongo}.
The case of $m=0$ and $\mathcal{O}$ being the forward lightcone can also be treated \cite{1978Buchholz}.

For a double cone in the massive case ($m > 0$), an explicit description of $M_\pm$ has been the subject of much investigation, but without conclusive results.
Recently it has been proposed \cite{2020LongoMorsella,2022Longo} that also for double cones, $M_-$ is actually a multiplication operator independent of $m$, and thus the result for $m=0$ can be employed.
However, since the proof there has a gap,\footnote{
The claim in \cite{2022Longo} that the operator in Eq.~(93) there is anti-Hermitian with respect to  the relevant \emph{complex} scalar product has turned out to be false.} we consider the problem still open.

In this paper, rather than attempting to make explicit guesses for $M_-$, we approach the problem of the massive modular generator for double cones with \emph{numerical} methods. 
That is, we evaluate formula \eqref{eq:ModularGenerator.Blocks} numerically, discretizing the integral kernels of the operators $\chi$ and $A$ in a suitable basis, and thus replacing the problem with one in matrix algebra. 
In this way, we find an approximation to the (integral kernels of the) operators $M_\pm$ for various parameters; in particular, we investigate the mass dependence of $M_-$.

To that end, we will first recall the formal framework of modular operators for one-particle structures in \autoref{sec:ModularOperatorOneParticle}. We then explain our numerical approach to evaluating the expression \eqref{eq:ModularGenerator.Blocks} in \autoref{sec:NumericalConcept}; we compare the approximation to existing results for $(1 + 1)$-dimensional wedge regions as our test case (\autoref{sec:D2RightWedge}). As our main result, we apply the same methods to the modular generator of a double cone in $1 + 1$ dimensions (\autoref{sec:D2DoubleCone}) and in $3 + 1$ dimensions (\autoref{sec:D4DoubleCone}). We end with a discussion and outlook in \autoref{sec:Conclusion}. The computer code used for producing the results is provided as supplemental material to this paper.

\section{Modular operators for one-particle structures}
\label{sec:ModularOperatorOneParticle}

In this paper, we study the modular operator of local algebras for linear quantum field theories in a ground state (vacuum) representation; these can be described on the ``one-particle level'' via second quantization.
The one-particle space, denoted $\mathcal{H}$ below, is the solution space of a second-order wave equation, here taken to be of the form $( \partial_t^2 + A ) \phi = 0$; it can be parametrized by two pieces of time-0 initial data, and is equipped with a suitable Hilbert space scalar product.
We briefly repeat the mathematical background.

\begin{definition}\label{def:ops}
 A \emph{one-particle structure} $(\mathcal{H}_\mathrm{r}, \mathcal{L}_\mathrm{r}, A)$ is given by a separable real Hilbert space $\mathcal{H}_\mathrm{r}$, a closed subspace $\mathcal{L}_\mathrm{r} \subset \mathcal{H}_\mathrm{r}$, and a positive (possibly unbounded, but densely defined) operator $A$ on $\mathcal{H}_\mathrm{r}$ with dense range, 
 with the following property:
 
   For any $s \in \Reals$, denote by $\overline{\phantom{m}}^s$ the closure in the norm $\| \cdot \|_s := \| A^s \cdot \|_{\mathcal{H}_\mathrm{r}}$, and set
	\begin{align}
				\mathcal{H}_\mathrm{r}^s
		&:= \overline{\dom A^s}^{s}
		,
&
				\mathcal{L}_\mathrm{r}^s
		&:= \overline{\mathcal{L}_\mathrm{r} \cap \dom A^s}^{s}
		\subset \mathcal{H}_\mathrm{r}^s
		,
&
				\mathcal{L}_\mathrm{r}^{\perp, s}
		&:= \overline{\mathcal{L}_\mathrm{r}^\perp \cap \dom A^s}^{s}
		\subset \mathcal{H}_\mathrm{r}^s
		,
	\end{align}
	where $\mathcal{L}_\mathrm{r}^\perp$ is the orthogonal complement of $\mathcal{L}_\mathrm{r}$. 
  Further, denote  by $\langle \cdot, \cdot\rangle_{s}$ the dual pairing between $\mathcal{H}_\mathrm{r}^s$ and $\mathcal{H}_\mathrm{r}^{-s}$ arising from the scalar product $\langle \cdot, \cdot\rangle$ in $\mathcal{H}_\mathrm{r}$, 
  and by $\vphantom{m}^\circ$ the polars for this dual pairing. Then we demand for $s =\pm \frac{1}{4}$ that
  \begin{subequations}
  \label{eq:SubspaceProperties}
	\begin{align}
	\label{eq:Lperpinter}
				\mathcal{L}_\mathrm{r}^{s}
			\cap \mathcal{L}_\mathrm{r}^{\perp ,s}
		&= \{ 0 \}
		;
	\\
	\label{eq:Lpolar}
				( \mathcal{L}_\mathrm{r}^{s} )^{\circ}
		&= \mathcal{L}_\mathrm{r}^{\perp,-s}
		;
	\\
	\label{eq:Lintersect}
				A^{s} \mathcal{L}_\mathrm{r}^{s}
			\cap A^{-s} \mathcal{L}_\mathrm{r}^{-s} 
		&= \{ 0 \}
		= A^{s} \mathcal{L}_\mathrm{r}^{\perp,s}
			\cap A^{-s} \mathcal{L}_\mathrm{r}^{\perp, -s}
		.
	\end{align}
  \end{subequations}

\end{definition}
\noindent
(Note here that $A^{t}$, $t \in \Reals$, naturally extends to a bounded operator from $\mathcal{H}_\mathrm{r}^s$ to $\mathcal{H}_\mathrm{r}^{s-t}$.)

The canonical example is the real scalar free field in time-0 formulation, where $\mathcal{H}_\mathrm{r} := \mathrm{L}^2_{\Reals}( \Reals^{d-1} )$ with $d \geq 2$ (with the usual real inner product), $A := -\Laplace + m^2$ with some $m>0$ (or $m \geq 0$ if $d \geq 3$), and $\mathcal{L}_\mathrm{r} := \mathrm{L}^2_\Reals( \mathcal{B} )$ for a suitable region $\mathcal{B} \subset \Reals^{d-1}$.
For the properties \eqref{eq:SubspaceProperties} in the case $d = 4$, see \cite[Sec.~2]{1989FiglioliniGuido}; the case $d=2$ is analogous.

Given a one-particle structure, we define the space $\mathcal{H} = \mathcal{H}_\mathrm{r}^{\isupfrac{1}{4}} \oplus \mathcal{H}_\mathrm{r}^{-\isupfrac{1}{4}}$ equipped with the complex structure
\begin{align}
			i_{A}
	&:=	\begin{pmatrix}
			0 & A^{-\supfrac{1}{2}} \\
			-A^{\supfrac{1}{2}} & 0
			\end{pmatrix}
	.
\end{align}
The complex-linear scalar product $( f, g \in \mathcal{H} )$
\begin{align}
			\innerProd[\mathcal{H}]{f}{g}
	&:= \innerProd[\subfrac{1}{4} \oplus -\subfrac{1}{4}]{f}{\left( A^{\supfrac{1}{2}} \oplus A^{-\supfrac{1}{2}} \right) g}
		+ \i \innerProd[\subfrac{1}{4} \oplus -\subfrac{1}{4}]{f}{\sigma g}
	,
&
			\sigma
	&= \begin{pmatrix}
			0 & 1 \\
			-1 & 0
		\end{pmatrix}
	,
\end{align}
makes $\mathcal{H}$ into a complex Hilbert space.

In $\mathcal{H}$, consider the subspace $\mathcal{L} := \mathcal{L}_\mathrm{r}^{\isupfrac{1}{4}} \oplus \mathcal{L}_\mathrm{r}^{-\isupfrac{1}{4}}$ and its symplectic complement $\mathcal{L}'$ w.r.t.\ the symplectic form $\Im\langle \cdot,\cdot \rangle_\mathcal{H}$. Straightforward computations using \eqref{eq:SubspaceProperties} yield (cf.~\cite[Prop.~2.7]{1989FiglioliniGuido}):
\begin{lemma}\label{lemma:ellprime}
 One has $\mathcal{L}' = \mathcal{L}_\mathrm{r}^{\perp, \isupfrac{1}{4}} \oplus \mathcal{L}_\mathrm{r}^{\perp, -\isupfrac{1}{4}}$.
 The subspace $\mathcal{L}\subset\mathcal{H}$ is standard ($\mathcal{L} \cap i_A \mathcal{L} = \{0\}$, $\overline{\mathcal{L} + i_A \mathcal{L}}=\mathcal{H}$) and factorial ($\mathcal{L} \cap \mathcal{L}' = \{0\}$).
\end{lemma}

Thus we can define the (one-particle) Tomita operator $T$ with respect to $\mathcal{L}$ as the closure of 
\begin{align}
			f + \i g
	&\mapsto f - \i g
&
			\text{for~} f, g
	&\in \mathcal{L}
	.
\end{align}
The polar decomposition of $T$, written as $T = J \Delta^{\isupfrac{1}{2}}$, is the object of interest in this paper. 

In quantum field theory, one is then interested in the Weyl algebra $\mathfrak{W}(\mathcal{H})$ in its Fock representation arising from $\Re \langle \cdot, \cdot \rangle_\mathcal{H}$, and in the subalgebra $\mathfrak{W}(\mathcal{L})\subset \mathfrak{W}(\mathcal{H})$. The Tomita-Takesaki modular operator of $\mathfrak{W}(\mathcal{L})$ with respect to the Fock vacuum is the second quantization $\Gamma(\Delta)$ of the one-particle modular operator $\Delta$ from above, for example, see \cite{1989FiglioliniGuido}. For the purposes of this article, the Fock representation will play no role, and we focus our attention on the one-particle structure.

Our aim is to find an explicit form of the modular operator $\Delta$ in terms of $A$ and $\mathcal{L}_\mathrm{r}$. This is in fact possible in terms of spectral calculus, 
as first shown in \cite{1989FiglioliniGuido} for the real scalar field. To that end, certain projectors onto $\mathcal{L}$ and $\mathcal{L}_\mathrm{r}^{\pm \isupfrac{1}{4}}$ play an important role. Denote by $P$ the projector in $\mathcal{H}$ with image $\mathcal{L}$ and kernel $\mathcal{L}'$; on $\dom P = \mathcal{L} + \mathcal{L}'$, it is closed but in general unbounded. Also, denote by $\chi$ the \emph{orthogonal} projector onto $\mathcal{L}_\mathrm{r}\subset\mathcal{H}_\mathrm{r}$. Thanks to \eqref{eq:Lperpinter}, for $s = \pm \frac{1}{4}$, we can likewise define the closed projectors $\chi_s$ in $\mathcal{H}_\mathrm{r}^s$ with image $\mathcal{L}_\mathrm{r}^s$ and kernel $\mathcal{L}_\mathrm{r}^{\perp,s}$, with $\dom \chi_{s} = \mathcal{L}_\mathrm{r}^s + \mathcal{L}_\mathrm{r}^{\perp,s}$. 
We now state:
\begin{proposition}
\label{prop:ModularFormula}
	For a one-particle structure $(\mathcal{H}_\mathrm{r},\mathcal{L}_\mathrm{r},A)$, let $B$ be the operator on $\mathcal{H}_\mathrm{r}$ given as
	\begin{subequations}
	\label{eq:coth}
	\begin{align}
			B
		&= A^{\supfrac{1}{4}} \chi_{\subfrac{1}{4}} A^{-\supfrac{1}{4}}
			+ A^{-\supfrac{1}{4}} \chi_{-\subfrac{1}{4}} A^{\supfrac{1}{4}}
			- 1
		,
	\\
			\dom B
		&= \left(
					A^{\supfrac{1}{4}}
					\left(
						\mathcal{L}_\mathrm{r}^{\supfrac{1}{4}}
					+ \mathcal{L}_\mathrm{r}^{\perp, \supfrac{1}{4}}
					\right)
				\right)
			\cap \left(
					A^{-\supfrac{1}{4}}
					\left(
						\mathcal{L}_\mathrm{r}^{-\supfrac{1}{4}}
					+ \mathcal{L}_\mathrm{r}^{\perp, -\supfrac{1}{4}}
					\right)
				\right)
			.
	\end{align}
	\end{subequations}
	$B$ is essentially self-adjoint, and denoting its closure by the same symbol, one has
	\begin{align}
	\label{eq:modGenFormula}
				\log \Delta
		&= i_A \begin{pmatrix}
					0 & M_- \\
					-M_+ & 0
				\end{pmatrix}
		,
	&
		\text{where}\quad
				M_\pm
		&= 2 A^{\pm \supfrac{1}{4}} \operatorname{arcoth}(B) A^{\pm \supfrac{1}{4}}
		.
	\end{align}
\end{proposition}
\begin{proof}
It has been shown in \cite[Theorem~2.2]{2020CiolliLongoRuzzi} that for the projector $P$,
\begin{equation}
\label{eq:pformula}
			P \restriction \mathcal{D}_0 = (1+T) \frac{1}{1-\Delta} \restriction \mathcal{D}_0,
\end{equation}
where the domain $\mathcal{D}_0$ contains all vectors of compact spectral support (w.r.t.~$\Delta$) not containing 0 and 1; and $\mathcal{D}_0$ is a core for $P$.
Now defining the complex-linear operator $Q := P - i_A P i_A -1$, first on $\mathcal{D}_0$, we have
\begin{equation}\label{eq:QDelta}
   Q \restriction \mathcal{D}_0 = \frac{1+\Delta}{1-\Delta} \restriction \mathcal{D}_0.
\end{equation}
Now clearly the r.h.s.~is essentially self-adjoint, and therefore $Q$ is. The domain of its self-adjoint closure contains at least
\begin{equation}
  \mathcal{D}_1 = \dom( P) \cap \dom (i_A P i_A) = (\mathcal{L} + \mathcal{L}') \cap (i_A\mathcal{L} + i_A\mathcal{L}') \supset \mathcal{D}_0.   
\end{equation}
Since relation \eqref{eq:QDelta} holds analogously for the closure of $Q$, denoted by the same symbol, we find by spectral calculus
\begin{equation}
\label{eq:logdelta}
   \log \Delta = - 2 \operatorname{arcoth} Q.
\end{equation}
To simplify this, we consider the complex Hilbert space $\hat{ \mathcal{H} }= \mathcal{H}_\mathrm{r} \oplus \mathcal{H}_\mathrm{r}$ equipped with the complex structure given by $\sigma$ above, and the standard complex scalar product arising from $\mathcal{H}_\mathrm{r}$ and $\sigma$. The map $U = A^{\isupfrac{1}{4}} \oplus A^{-\isupfrac{1}{4}}: \mathcal{H} \to \hat{\mathcal{H}}$ is complex-linear and unitary.
Due to \autoref{lemma:ellprime}, we have $P=\chi_{1/4}\oplus \chi_{-1/4}$; thus
\begin{align}
			U P U^{-1}
	&= A^{\supfrac{1}{4}} \chi_{\subfrac{1}{4}} A^{-\supfrac{1}{4}}
			\oplus A^{-\supfrac{1}{4}} \chi_{-\subfrac{1}{4}} A^{\supfrac{1}{4}}. 
\end{align}
From there one obtains
\begin{align}
			U \mathcal{D}_1
	&= \dom B \oplus \dom B
	,
&
			U Q U^{-1} \restriction U \mathcal{D}_1
	&= B \oplus B
	.
\end{align}
Thus $B$ is essentially self-adjoint (since $Q$ is); and 
inverting $U$ and inserting into \eqref{eq:logdelta} gives the proposed result.  
\end{proof}

\section{Numerical approach}
\label{sec:NumericalConcept}

Our overall approach to find a concrete form of the modular generator $\log \Delta$ is to approximate an infinite-dimensional one-particle structure $(\mathcal{H}_\mathrm{r}, \mathcal{L}_\mathrm{r}, A)$, as arising in quantum field theory, with a sequence of \emph{finite-dimensional} ones $(\mathcal{H}_\mathrm{r}^{(n)}, \mathcal{L}_\mathrm{r}^{(n)}, A^{(n)})$, in such a way that the corresponding operators $\log \Delta^{(n)}$, or $M_\pm^{(n)}$, approximate $\log \Delta$ and $M_\pm$. 
Roughly speaking, we would choose suitable finite-dimensional orthogonal projectors $P^{(n)}$ in $\mathcal{H}_\mathrm{r}$, and define our ``discretized'' objects as 
\begin{align}
\label{eq:DiscConcept}
			\mathcal{H}_\mathrm{r}^{(n)}
	&:= P^{(n)} \mathcal{H}_\mathrm{r}
	,
&
			\mathcal{L}_\mathrm{r}^{(n)}
	&:= P^{(n)} \mathcal{L}_\mathrm{r}
	,
&
			A^{(n)}
	&:= \left( P^{(n)} A^{-s} P^{(n)} \right)^s
\end{align}
for some $s>0$; the reason for discretizing inverse powers of $A$ will become clear below.
In the finite-dimensional setting, the topological closures in \autoref{def:ops} do not play a role, and the $\chi_s^{(n)}$ all equal the orthogonal projector $\chi^{(n)}$.
Choosing an orthonormal basis $\left\{ e_j^{(n)} \right\}$ of $\mathcal{H}_\mathrm{r}^{(n)} = \mathcal{L}_\mathrm{r}^{(n)} + \mathcal{L}_\mathrm{r}^{\perp (n)}$, the problem of evaluating $B^{(n)}$ and $M_\pm^{(n)}$ according to \eqref{eq:modGenFormula} then reduces to a problem in matrix computation, which we treat with the usual methods of numerical linear algebra, including a numerical eigendecomposition of $B^{(n)}$ by Householder transformations and the QR algorithm (see, e.g., \cite[Ch.~8]{2013GolubVanLoan}).

In the present article, we make no attempt at rigorously establishing the convergence of $\log \Delta^{(n)}$ to $\log \Delta$; rather we show (in \autoref{sec:D2RightWedge}) that the numerical approximation scheme gives reasonable results in a well known example. However, let us add some heuristic remarks.

In our applications, $\mathcal{H}_\mathrm{r}$ will normally be a real-valued $\mathrm{L}^2$-space with some measure $\d{\mu}$. Our task is therefore to approximate the integral kernel $M_\pm( x, y )$ of the operator $M_\pm$. In certain examples, $M_-$ is a multiplication operator while $M_+$ is a second-order differential operator. As the kernel of $M_-$ is hence expected to be less singular in general, we focus our attention on this case; since $M_-$ determines $M_+$ via $M_+ = A^{\isupfrac{1}{2}} M_- A^{\isupfrac{1}{2}}$, this operator still contains all required information.

However, even for this example of $M_-$, its kernel is not a smooth function, but a distribution concentrated on the diagonal. Hence we \emph{cannot} expect that
\begin{equation}
  M_-^{(n)} (x,y):= \sum_{j,k} e_j^{(n)}(x)  \innerProd{ e_j^{(n)} }{ M_-^{(n)} e_k^{(n)} } e_k^{(n)}(y) \xrightarrow{n \to \infty} M_-(x,y)
\end{equation}
in the sense of pointwise convergence. Rather, we can expect this relation to be true in the weak sense: For sufficiently regular test functions $h$, $h'$, we should have
\begin{equation}\label{eq:weakConv}
    \sum_{j,k} \innerProd{ h}{ e_j^{(n)} } \innerProd{ e_j^{(n)} }{ M_-^{(n)} e_k^{(n)} } \innerProd{ e_k^{(n)}}{ h'} 
     \xrightarrow{n \to \infty} \iint h(x) M_-( x, y ) h'(y) \id{\mu( x )} \id{\mu( y )}.
\end{equation}
This point does show up in the numerical results, and will become important in our quantitative comparisons later on.

Another, somewhat technical remark is in order: For the numerical evaluation of the $\operatorname{arcoth}$ function in \eqref{eq:coth}, it is of course required that the spectrum of the discretized operator $B^{(n)}$ falls into $(-\infty,-1)\cup (1,\infty)$.
As long as the finite-dimensional structure still fulfills all conditions of \autoref{def:ops}, this is indeed the case, as e.g.~our computation in \autoref{sec:ModularOperatorOneParticle} shows.
However, in practice it turns out that the eigenvalues of $B^{(n)}$ are \emph{extremely} close to $\pm 1$, to the extent that round-off errors in the usual floating point precision lead to eigenvalues outside the allowed range.
A similar problem has been observed in the numerical studies of entanglement Hamiltonians in \cite[Sec.~3]{2022JaverzatTonni}.
For our results, we circumvent this problem by using an increased floating point precision of 450 decimal digits in two space-time dimensions and 640 digits in four dimensions.

We now need to choose suitable basis vectors $e_j^{(n)}$ in concrete examples of the one-particle structure $(\mathcal{H}_\mathrm{r},\mathcal{L}_\mathrm{r},A)$
such that \emph{both} $\mathcal{L}_\mathrm{r}^{(n)}$ and $A^{(n)}$ are good approximations of their infinite-dimensional counterparts. We will first explain this choice, and the numerical results, in the well known example of a wedge region for a massive scalar field in $1 + 1$ dimensions (\autoref{sec:D2RightWedge}). We then make suitable modifications to find the modular generator for double cone regions in $1 + 1$ and $3 + 1$ dimensions (Sections~\ref{sec:D2DoubleCone} and \ref{sec:D4DoubleCone}).

\section{Test case: The wedge}
\label{sec:D2RightWedge}

In this section, we apply our numerical approach to a scalar field on $(1 + 1)$-dimensional Minkowski space with respect to the right wedge region,
$\mathcal{O} = \{ ( t, x ) \in \Reals^2 \mid	x > \lvert t \rvert \}$.
That is, in the language of \autoref{sec:ModularOperatorOneParticle}, we take
$\mathcal{H}_\mathrm{r} := \mathrm{L}^2_{\Reals}( \Reals )$ and its subspace
\begin{align}
			\mathcal{L}_\mathrm{r}
	&:= \left\{ f \in \mathcal{H}_\mathrm{r} \bypred \supp f \subset [0,\infty) \right\}
	,
\end{align}
together with the Helmholtz operator, $A:= - \partial_{x}^2 + m^2$ for some $m>0$.
In this situation, as mentioned, an explicit form of $\log \Delta$ is known \cite{1975BisognanoWichmann,2020CiolliLongoRuzzi}; namely, $M_-$ multiplies with the function $2\pi x$, or in terms of integral kernels,
\begin{equation}\label{eq:exactWedge}
  M_-( x, y ) = 2 \pi x \updelta( x - y ).
\end{equation}
This allows for a consistency check of our numerical results.

\subsection{Discretisation of the Hilbert space}

For discretizing the Hilbert space $\mathcal{H}_\mathrm{r}$, we use a basis of normalized rectangular functions, also called box functions.
To that end, let $i$ run over an integer interval $\{ 0, \ldots, n - 1 \}$ for some discretization size $n \in \Naturals$,
and choose points $a_{0} < \ldots < a_{n}$, $b_{i} := a_{i + 1}$.
We then choose the orthogonal basis functions $e_{i}^{( n )}$ to be supported in $[ a_{i}, b_{i} ] \subset \Reals$, namely, 
\begin{align}
\label{eq:BoxBasis.Element}
			e_{i}^{( n )}( x )
	&= n_{i}
			\upTheta( x - a_{i} )
			\upTheta( b_{i} - x )
	.
\end{align}
where the normalization factor $n_{i} > 0$ is chosen so that ${\innerProd{e_{i}^{( n )}}{e_{i}^{( n )}}} = 1$. 
In the present case, we actually choose linearly spaced grid points $a_i = -b + \frac{2 i}{n} b$ with some fixed cut-off $b > 0$, see \autoref{fig:BoxBasis10} for an example; 
we will label this basis as $e_{i}^{( n, b )}$. 
In later sections, we will also consider nonlinearly spaced grids. 
\begin{figure}
	\centering
	\includegraphics{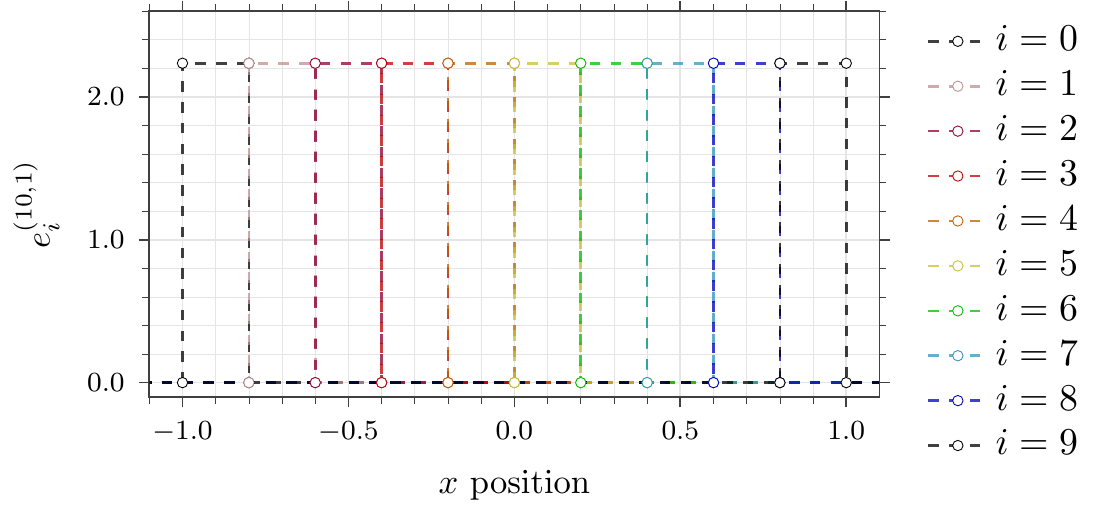}
	\caption{\label{fig:BoxBasis10} Basis of $n = 10$ box functions $e_{i}^{( n, b )}$ for a discretization over the interval $[ -b, b ]$ with $b = 1$. In this example, each element of the basis has the same width and hence the same normalization.}
\end{figure}

Choosing $n$ even, these functions $e_{i}$ are either contained in $\mathcal{L}_\mathrm{r}$ or $\mathcal{L}_\mathrm{r}^\perp$, hence the projector $\chi$ is simple to discretize in this basis; moreover, we have achieved $\dim \mathcal{L}_\mathrm{r}^{( n, b )} = \dim \mathcal{L}_\mathrm{r}^{\perp ( n, b )}$. 

It is perhaps less obvious whether this basis is chosen suitably to accommodate the discretization of $A$ and its positive and negative powers. In fact, one may notice that the functions $e_i^{( n, b )}$ are contained in the form domain, but not in the operator domain of the (undiscretized) operator $A^{\isupfrac{1}{4}}$. In order to investigate possible problems at this point, we alternatively used an orthonormal basis of continuous, piecewise linear functions; these are smoother and hence more adapted to the operators $A^{\pm \isupfrac{1}{4}}$, but the discretization of $\mathcal{L}_\mathrm{r}$ yields further complications in this case. Since this approach did not lead to qualitatively different numerical results compared with the box basis above, we do not report details here.
The discretization of $A^{\pm \isupfrac{1}{4}}$ in the box basis will be described in more detail in the next subsection.

\subsection{Discretisation of the Helmholtz operator}
\label{sec:ModHelmholtzD2}

In the present situation, we prefer to discretize $A^{-\isupfrac{1}{4}}$ rather than $A$ in the sense of \eqref{eq:DiscConcept}, since $A^{-\isupfrac{1}{4}}$ has an integral kernel of local class $\mathrm{L}^1$.
In fact, in terms of the modified Bessel function of the second kind $\BesselK{-\isubfrac{1}{4}}$ (and the Gamma function $\upGamma$) one has 
\begin{align}
\label{eq:ModHelmholtzD2.m1o4.Kernel}
	    A^{-\supfrac{1}{4}}( x, y )
	&= \left( \sqrt{\pi}\, \upGamma\bigl( \tfrac{1}{4} \bigr) \right)^{-1}
	    \left( \frac{2 m}{\lvert x - y \rvert} \right)^{\supfrac{1}{4}}
	    \BesselK{-\subfrac{1}{4}}\bigl( m \lvert x - y \rvert \bigr)
\end{align}
as the Fourier transform of $(p^2 + m^2)^{-\frac{1}{4}}$, see \cite[Table~1.3, Formula~(7)]{1954Bateman}.
This is a convolution kernel of the form $A^{-\isupfrac{1}{4}}( x, y ) = f( \lvert x - y \rvert )$, with $f$ integrable on $[ 0, \infty )$ and smooth on $( 0, \infty )$.
Its matrix elements in the box basis are hence computed as
\begin{align}
\label{eq:BoxBasis.ModHelmholtzD2.m1o4}
			\left( A^{-\supfrac{1}{4}} \right)^{( n, b )}_{i j}
	&:= \innerProd{e_{i}^{( n, b )}}{A^{-\supfrac{1}{4}} e_{j}^{( n, b )}}
  = n_{i} n_{j}
  		\iint_{[ a_{i}, b_{i} ] \times [ a_{j}, b_{j} ]}
        f\bigl( \lvert x - y \rvert \bigr)
      \id{x} \id{y}
  .
\end{align}
Using the substitution $( x', y' ) = ( y + x - a_{i} - a_{j}, y - x )$, and splitting the integration region into parts $x \geq y$ and $x \leq y$ where necessary, one finds for $i<j$,
\begin{subequations}
\begin{align}
\label{eq:BoxBasis.ModHelmholtzD2.m1o4.Diagonal}
			\left( A^{-\supfrac{1}{4}} \right)^{( n, b )}_{i i}
  &= 2 n_{i}^2 F( b_{i} - a_{i} )
  ,
\\
			\left( A^{-\supfrac{1}{4}} \right)^{( n, b )}_{i j}
	&= \left( A^{-\supfrac{1}{4}} \right)^{( n, b )}_{j i}
\\\label{eq:BoxBasis.ModHelmholtzD2.m1o4.NonDiagonal}
  &= n_{i} n_{j}
  		\bigl(
      	F( b_{j} - a_{i} )
      - F( b_{j} - b_{i} )
  		- F( a_{j} - a_{i} )
      + F( a_{j} - b_{i} )
			\bigr)
	,
\end{align}
\end{subequations}
with
\begin{align}
\label{eq:BoxBasis.ModHelmholtzD2.Antiderivative}
			F( x )
	&:= x \int_0^x f( y' ) \id{y'} - \int_0^x y' f( y' ) \id{y'}
	.
\end{align}
Using \eqref{eq:ModHelmholtzD2.m1o4.Kernel}, the expression for $F( x )$ can be written explicitly in terms of generalized hypergeometric functions.

As matrix representation of $A^{\isupfrac{1}{4}}$, we use the numerically computed matrix inverse of $A^{-\isupfrac{1}{4}}$; due to the peculiar nature of the spectrum of $B^{( n, b )}$ as mentioned in \autoref{sec:NumericalConcept}, it is important that these two matrices are, as close as possible within floating point precision, inverses of each other.

All further steps in the calculation are performed with the respective matrices and their eigenvalue decomposition, obtained by standard algorithms.

\subsection{Numerical results}
\label{sec:D2RightWedge.Results}

Now proceeding to the results of our numerical scheme, let us compare the matrix $\log \Delta^{( n, b )}$ or rather $M_-^{( n, b )}$ to the kernel of the ``exact solution'' $M_-$ in \eqref{eq:exactWedge}, in dependence of the discretization parameters $n$ and $b$. We first set the model parameter $m=1$, but will come back to more general values of $m$ later.

\begin{figure}
	\centering
	\includegraphics{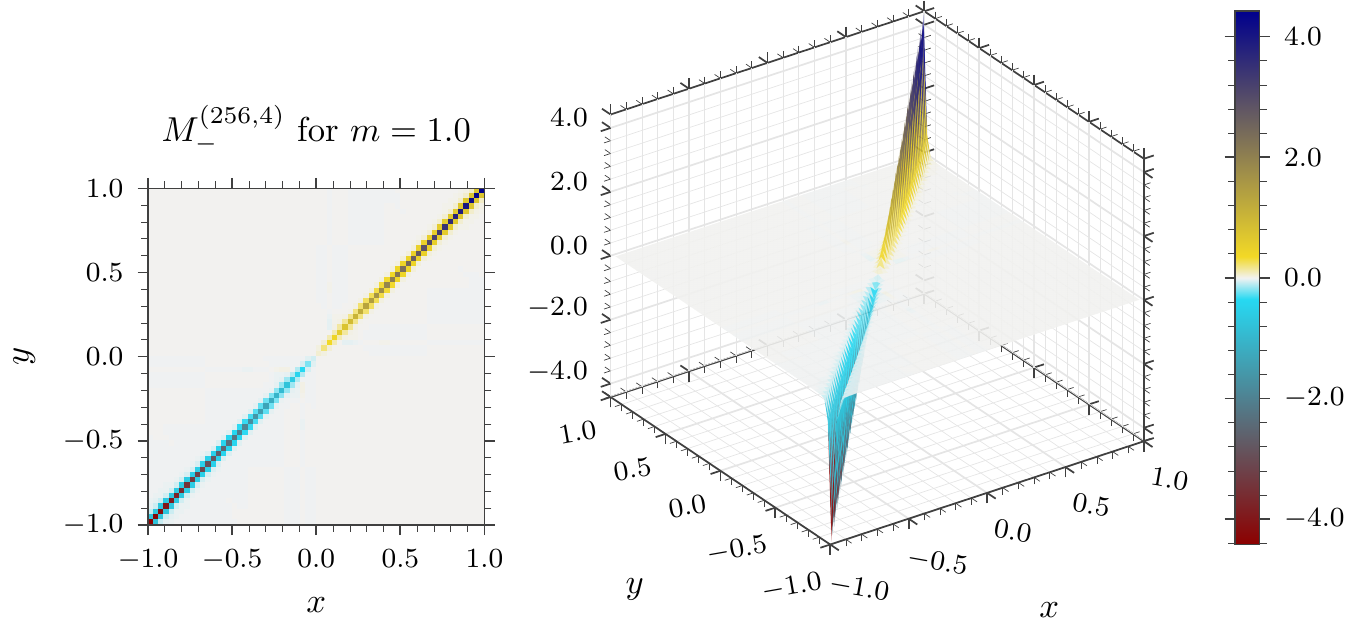}
	\caption{\label{fig:D2RightWedge.kernelMm} Example of the operator kernel $M_-^{( n, b )}( x, y )$ discretized by $n = 256$ box functions over an equally spaced grid for the range $[ -4, 4 ]$ and a mass parameter $m = 1.0$. The matrix is almost diagonal (see matrix plot on the left) with values falling off very rapidly away from the diagonal (see surface plot on the right). Both plots share the same colour scale from dark red for large negative values, through cyan for small negative values, very light grey at 0, yellow for small positive values, to dark blue for large positive values.}
\end{figure}
First, we show the result for $M_-^{( n, b )}$ with parameters $n = 256, b = 4$ in \autoref{fig:D2RightWedge.kernelMm}. As expected, the matrix entries are very small except near the diagonal $i=j$. However, while decaying fast with $\lvert i - j \rvert$, the entries certainly still have a noticeable magnitude, e.g., for $\lvert i - j \rvert = 1$, even if the ``exact solution'' is proportional to a Dirac delta. This is no less than expected in a numerical approximation, but emphasizes our point that convergence to the undiscretized operator needs to be read in the weak sense, see~\eqref{eq:weakConv}; for a quantitative comparison, we need to ``smear'' the result with suitable test functions.

To that end, we choose a set of Gaussian functions with a fixed width (standard deviation) $\sigma$ that extends over a small number of grid points and a varying position parameter $\mu_{i}$, 
\begin{align}
\label{eq:GaussianTestFunction}
			h_{i}( x )
	&= \frac{1}{\sqrt[4]{\pi \sigma^2}}
			\exp\left( -\frac{( x - \mu_{i} )^2}{2 \sigma^2} \right)
	.
\end{align}
These are normalized with respect to the inner product of $\mathcal{H}_\mathrm{r}$. 
For the smearing to be effective, the width $\sigma$ has to be chosen somewhat larger than the grid spacing (even for the lowest resolutions $n$ in the comparison).
On the other hand, substantially larger $\sigma$, while valid, would obscure detail in the comparisons below.

Focussing on values near the diagonal, we now evaluate both sides of \eqref{eq:weakConv} for $h=h'=h_i$. For the expectation values of the multiplication operator \eqref{eq:exactWedge}, we find
\begin{equation}
\label{eq:D2RightWedge.ModularOperator.UpperRightBlock.Reference.Smeared}
			\innerProd{h_{i}}{M_- h_{i}}
	= 2 \pi \mu_{i}
	.
\end{equation}
Also, let $h_{i k} = {\innerProd{h_i}{e_k^{( n, b )}}}$; these can be computed in terms of Gaussian error functions or by numerical integration.
Then our numerical approximation to \eqref{eq:D2RightWedge.ModularOperator.UpperRightBlock.Reference.Smeared} is
\begin{equation}
\label{eq:ModularOperator.UpperRightBlock.Smeared.Matrix}
    \innerProd{h_{i}}{M_- h_{i}}
	\approx \sum_{k, l =0}^{n-1}
				h_{i k}
				\left( M_-^{( n, b )} \right)_{k l}
				h_{i l}
	.
\end{equation}

For the plots, we let $0 \leq i \leq 40$, $\sigma = \frac{6}{32}$, and $\mu_{i}$ range over $[ -2, 2 ]$ equally spaced.
\begin{figure}
	\centering
	\includegraphics{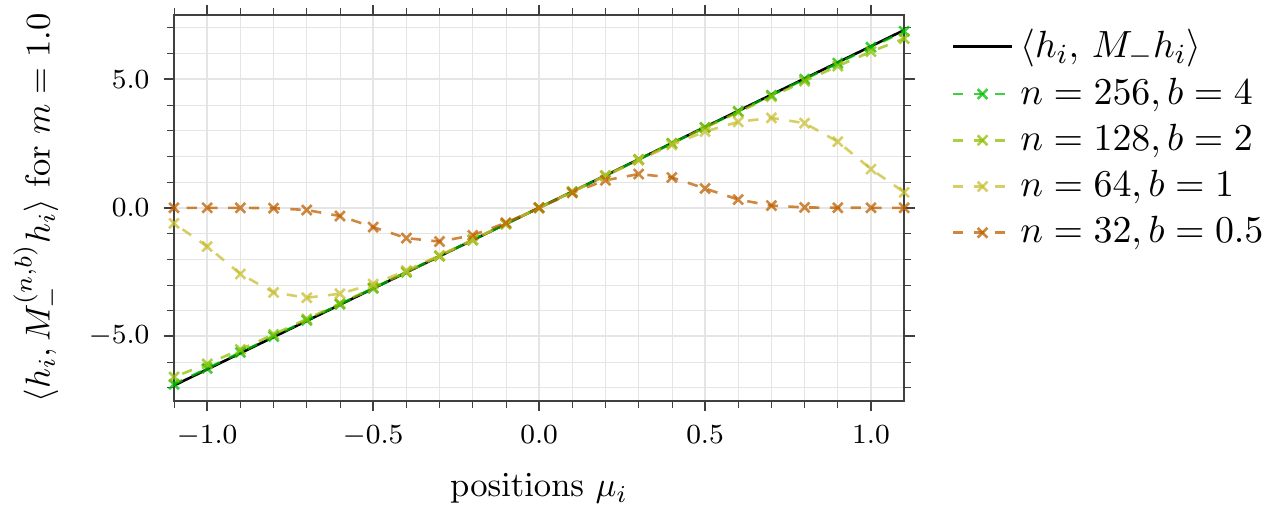}
	\caption{\label{fig:D2RightWedge.Mm.disc} Main diagonal of $M_-^{( n, b )}$ for the right wedge with a mass parameter $m = 1.0$, smeared against test functions \eqref{eq:GaussianTestFunction} at different positions $\mu_{i}$. Note that in some cases, not the full discretization interval $[ -b, b ]$ is shown.}
\end{figure}
\autoref{fig:D2RightWedge.Mm.disc} shows the results for different discretization sizes $n$ and different discretization intervals $[ -b, b ]$ such that the width of the box functions remains constant, $\frac{2 b}{n} = \frac{1}{32}$.
The numeric results (dashed lines with crosses, varying colour) are calculated with 450 decimal places of precision, and they approximate the expected linear expression increasingly better on the fixed interval $[-1,1]$, in the sense that errors caused by ``boundary effects'' are less noticeable in this interval as $b$ is increased.
To demonstrate this further, consider the relative error
\begin{align}
\label{eq:D2RightWedge.ModularOperator.UpperRightBlock.Smeared.MatrixError}
			\mathrm{err}_{i}\left( M_-^{( n, b )} \right)
	&:= \left\lvert 
				1
			- \frac{%
					\sum_{k, l}
						h_{i k}
						( M_-^{( n, b )} )_{k l}
						h_{i l}
				}{%
					\innerProd{h_{i}}{M_- h_{i}}
				}
			\right\rvert
	.
\end{align}
These values are shown, for different discretization sizes $n$ and a fixed boundary parameter $b = 4$, in \autoref{fig:D2RightWedge.Mm.discn.error}.
\begin{figure}
	\centering
	\includegraphics{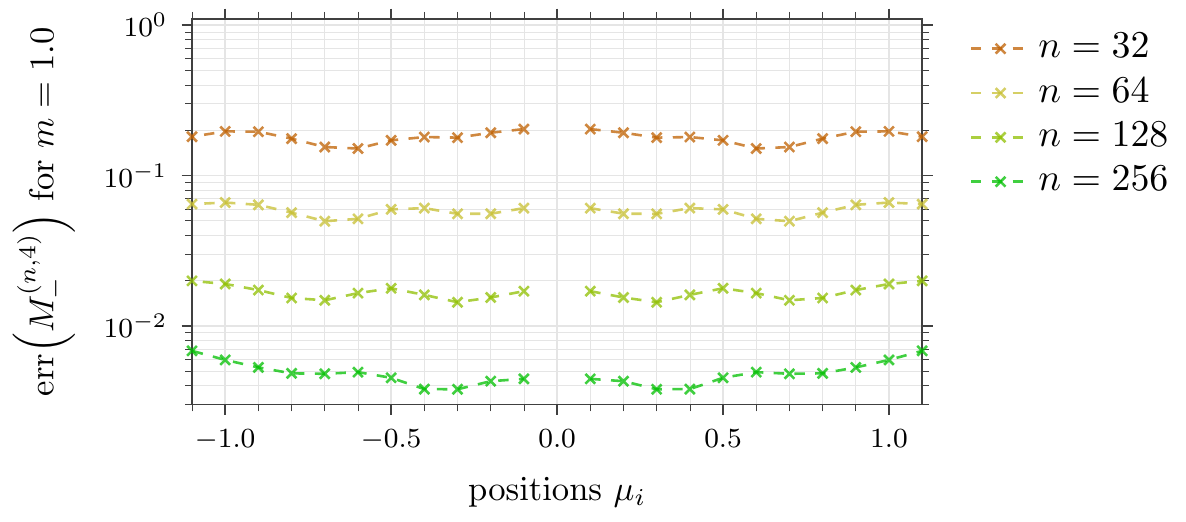}
	\caption{\label{fig:D2RightWedge.Mm.discn.error} Behaviour of the relative error \eqref{eq:D2RightWedge.ModularOperator.UpperRightBlock.Smeared.MatrixError} (on a logarithmic scale) for increasing discretization sizes $n$ while keeping the discretization range $[ -4, 4 ]$ fixed.}
\end{figure}
The relative error is approximately reduced by a constant factor each time the resolution is doubled. 

It is evident from the graphs that the cut-off at $\pm b$ does introduce a noticeable error in the result -- in fact, the contribution of matrix elements ``near the boundary'' in $B^{( n, b )}$ and $M_-^{( n, b )}$ is by no means small -- but that nevertheless, we obtain a good approximation of the integral kernel of $M_-$ in the region sufficiently far \emph{away} from the boundary.  

Since the numerical results for $n = 256$ over the discretization range $[ -4, 4 ]$ are sufficiently close, let us now take a look at mass independence.
We vary the mass parameter $m$ to obtain the results shown in \autoref{fig:D2RightWedge.Mm.mass}.
\begin{figure}
	\centering
	\includegraphics{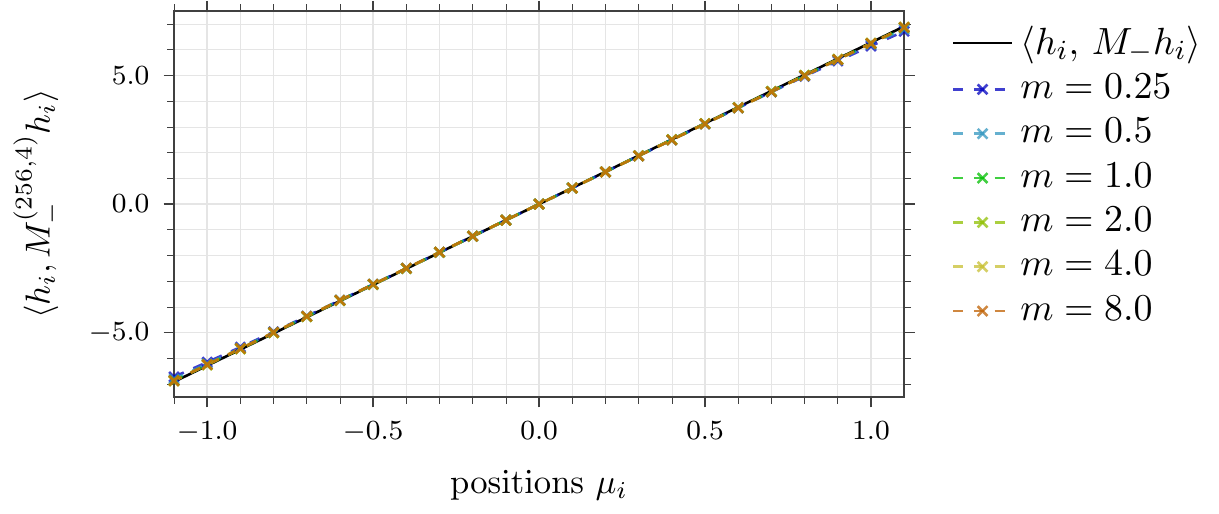}
	\caption{\label{fig:D2RightWedge.Mm.mass} Main diagonal of $M_-^{( n, b )}$ for the right wedge, discretized over the range $[ -4, 4 ]$ with $n = 256$ box functions, and then smeared against Gaussian functions \eqref{eq:GaussianTestFunction}. The results are shown for a varying mass parameter, but the numeric data points (dashed lines with crosses) cover the analytic result (solid, black line).}
\end{figure}
As expected from the exact result, there is no noticeable mass dependence, up to numerical errors due to the finite discretization and effects closer to the boundary of the discretization interval.

\section{The double cone in 1 + 1 dimensions}
\label{sec:D2DoubleCone}

Now that we have tested our numerical approach against the analytic solution for the right wedge in two-dimensional Minkowski spacetime, let us consider examples for which explicit expressions for the modular operator are yet unknown. 

Our first example is, again in the $(1 + 1)$-dimensional scalar field, the subspace for the double cone, namely the causal closure of the interval $[ -1, 1 ]$ in the time-0 plane. 
That is, we set as in the wedge case
\begin{align}
			\mathcal{H}_\mathrm{r}
	&= \mathrm{L}^2_\Reals( \Reals )
	,
&
			A
	&= -\partial_x^2 + m^2
	,
\end{align}
but now
\begin{align}
			\mathcal{L}_\mathrm{r}
	&= \left\{ f \in \mathcal{H}_\mathrm{r} \bypred \supp f \subset [ -1, 1 ] \right\}
	.
\end{align}
(Since the subspaces for double cones with other centres or radii can be unitarily mapped to the above situation with an appropriate change in $m$, there is no loss in generality in choosing the interval $[-1,1]$.)

The discretization of the Hilbert space $\mathcal{H}_\mathrm{r}$ and the operator $A$, as well as the following computation, are handled essentially as in \autoref{sec:D2RightWedge}. However, in choosing the box basis $\{e_j^{( n, b )}\}$, we make the following change: in order to keep $\dim \mathcal{L}_\mathrm{r} = \dim\mathcal{L}_\mathrm{r}^\perp$, we choose the grid points $a_i$ equally spaced only in the interval $[ -1, 1 ]$, and consider three different grid spacings outside the interval.
For $b = 2$, the outside is also equally spaced, while for $b = 4$ and $b = 6$, we choose a spacing that increases linearly towards the cut-offs at $\pm b$ starting from the fixed value of the inner spacing such that a quarter of the basis functions is supported to the left, half of them supported inside, and another quarter supported to the right of the interval. 
The discretization size is set to $n = 256$; the results for $n = 128$ are coarser but very similar, hence we do not include them here.

\begin{figure}
	\centering
	\includegraphics{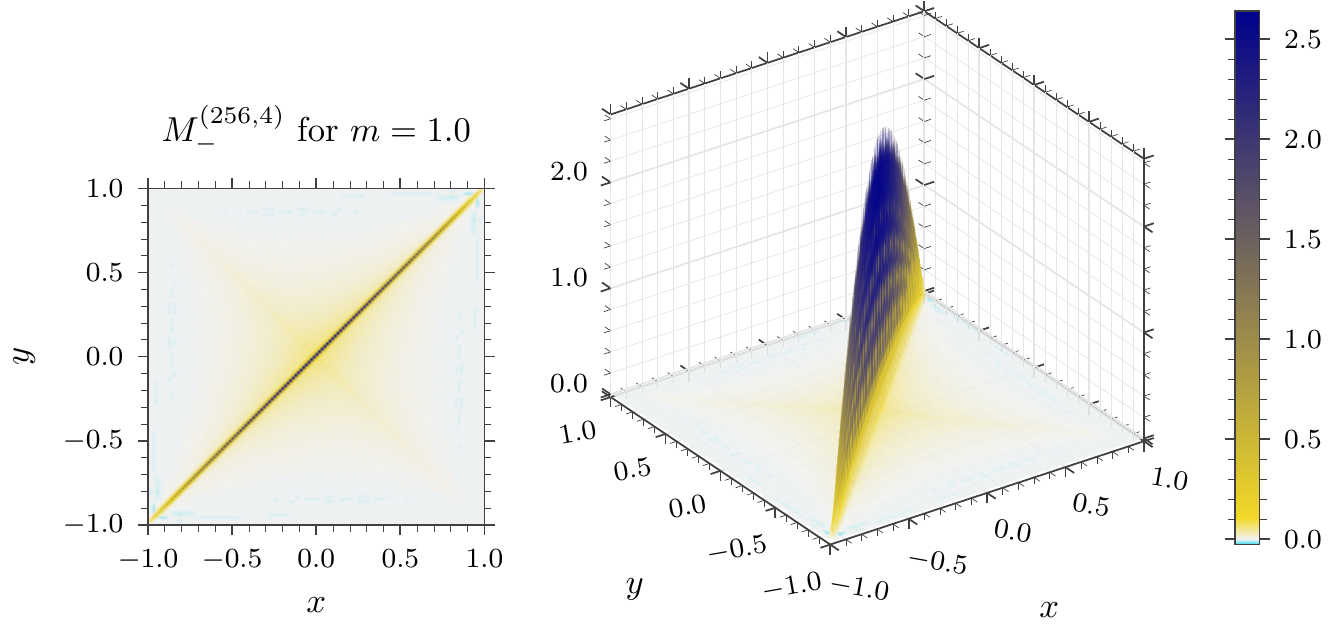}
	\caption{\label{fig:D2DoubleCone.kernelMm} Example of the discretization of the operator kernel $M_-^{( n, b )}( x, y )$ (at $m = 1.0$) in $1 + 1$ dimensions. The discretization uses $n = 256$ box functions over the range $[ -4, 4 ]$ with half of the functions supported on $[ -1, 1 ]$ (equally spaced grid), and increasingly larger grid steps away from the interval $[ -1, 1 ]$ (not visible in the plots). Similar to \autoref{fig:D2RightWedge.kernelMm}, the kernel is concentrated near the diagonal (see matrix plot on the left) and falls off strongly away from the diagonal (see surface plot on the right). Both plots share the same colour scale from very light grey at 0, through yellow for small positive values, to dark blue for large positive values.}
\end{figure}
An example result for the matrix $M_-^{( n, b )}$ with a basis of 256 box functions over the interval $[ -4, 4 ]$, and with $m = 1$, is shown in \autoref{fig:D2DoubleCone.kernelMm}; we focus here on the interior of the interval. The matrix is supported along the diagonal (first plot) and the values away from the diagonal fall off rapidly (second plot). By these results, $M_-$ still appears to be at least close to a multiplication operator.

\begin{figure}
	\centering
	\includegraphics{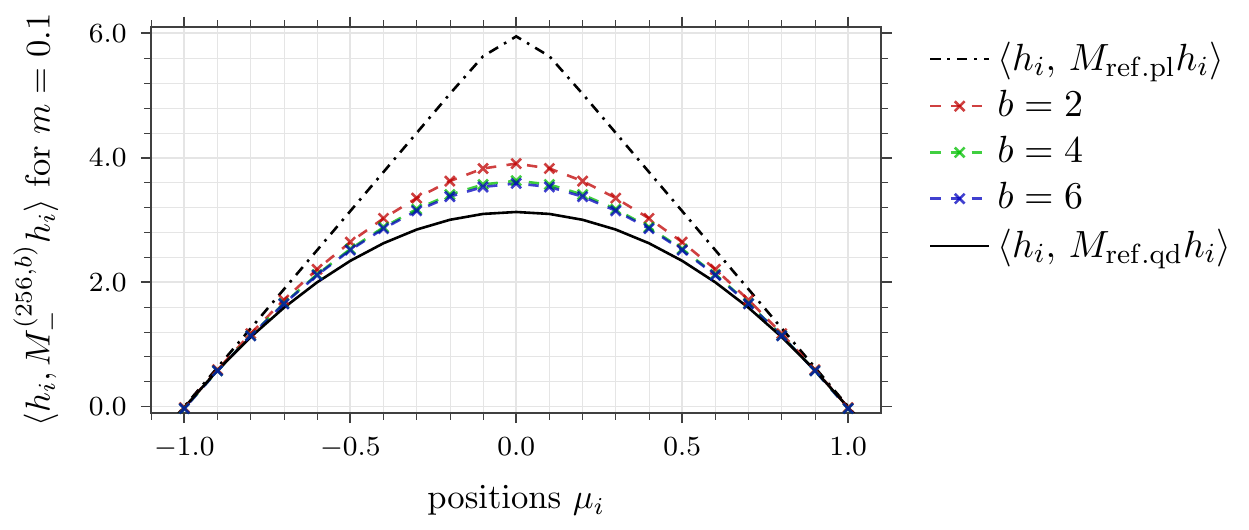}
	\caption{\label{fig:D2DoubleCone.Mm.disc} Comparison of $M_-^{( n, b )}$ for the $(1 + 1)$-dimensional double cone, smeared against Gaussian functions $h_{i}$ as in \autoref{fig:D2RightWedge.Mm.disc}, for different discretization ranges $[ -b, b ]$ but with a fixed resolution $n = 256$.}
\end{figure}
A comparison of $M_-^{( n, b )}$ for a fixed mass $m = 0.1$ but a varying discretization parameter $b \in \{ 2, 4, 6 \}$ is shown in \autoref{fig:D2DoubleCone.Mm.disc}.
Even though all discretizations have the same grid and hence the same number of basis functions supported within the interval, the difference choices of an outer grid affect the results in the interior, with the largest difference between $b = 2$ and $b = 4$.
Since the cut-off effects become very small when increasing the discretization range further from $b = 4$ to $b = 6$, in the following we report the results for $b = 4$ only.
Notice that the non-equal grid spacing yields twice the resolution over the interval region when compared to the case of the right wedge.
Hence we decrease the width of the Gaussian test functions to $\sigma = \frac{6}{64}$, but otherwise use the same set of functions as in \autoref{sec:D2RightWedge}.

Let us now quantitatively investigate the behaviour of the values near the diagonal when changing the mass parameter.
We consider two reference values: The first is the quadratic result (short: qd) for $m = 0$ that is expected for higher dimensions \cite{1982HislopLongo},
\begin{align}
\label{eq:D2DoubleCone.ModularOperator.UpperRightBlock.ReferenceQD}
			M_{\mathrm{ref.qd}}( x, y )
	&:= \pi ( 1 - x^2 ) \updelta( x - y )
	.
\end{align}
For the second, note that our double cone is the intersection of a left wedge with tip at $x = 1$ and a right wedge with tip at $x = -1$. Denoting by $\Delta_{\mathrm{L}},\Delta_\mathrm{R}$ their associated modular operators, we have $-\log \Delta \leq -\log \Delta_{\mathrm{L},\mathrm{R}}$ (cf.~\cite[Sec.~2.1]{2000Borchers}, with operator monotonicity of the logarithm). Thus $\innerProd{h}{M_- h} \leq \min_{\mathrm{D}\in\{\mathrm{L},\mathrm{R}\}}\innerProd{h}{M_{-,\mathrm{D}} h}$ from Prop.~\ref{prop:ModularFormula}, where $M_{-,\mathrm{D}}$ are known explicitly, analogous to \eqref{eq:exactWedge}. Assuming that $M_-$ for the double cone is still a multiplication operator, it must therefore be bounded above by the piecewise linear (short: pl) kernel
\begin{align}
\label{eq:D2DoubleCone.ModularOperator.UpperRightBlock.ReferencePL}
			M_{\mathrm{ref.pl}}( x, y )
	&:= 2 \pi
			\min\bigl(
			- ( x - 1 ),
			x + 1
			\bigr)
			\updelta( x - y )
	.
\end{align}
In expectation values for the Gaussian test functions \eqref{eq:GaussianTestFunction}, we obtain
\begin{subequations}
\label{eq:D2DoubleCone.ModularOperator.UpperRightBlock.References.Smeared}
\begin{align}
\label{eq:D2DoubleCone.ModularOperator.UpperRightBlock.ReferenceQD.Smeared}
			\innerProd{h_{i}}{M_{\mathrm{ref.qd}} h_{i}}
	&= \pi \left( 1 - \frac{\sigma^2}{2} - \mu_{i}^2 \right)
	,
\\
\label{eq:D2DoubleCone.ModularOperator.UpperRightBlock.ReferencePL.Smeared}
			\innerProd{h_{i}}{M_{\mathrm{ref.pl}} h_{i}}
	&= 2 \pi
			\biggl( 1 - \mu_{i} \erf\left( \frac{\mu_{i}}{ \sigma} \right) \biggr)
		- 2 \sigma \sqrt{\pi}
			\exp\left( -\frac{\mu_{i}^2}{\sigma^2} \right)
	.
\end{align}
\end{subequations}

\begin{figure}
	\centering
	\includegraphics{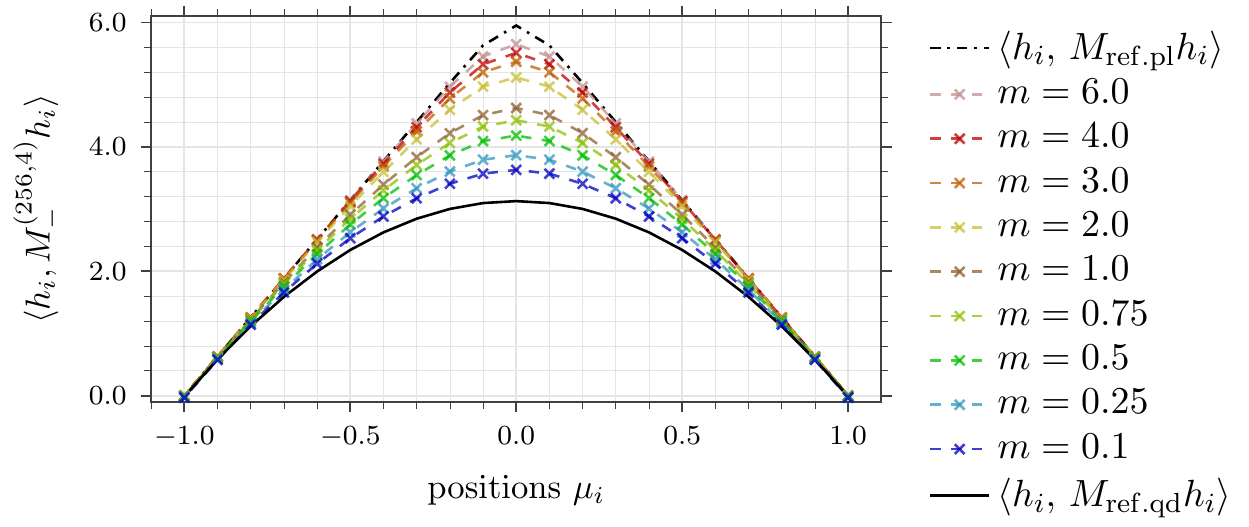}
	\caption{\label{fig:D2DoubleCone.Mm.mass} Comparison of the smeared operator kernel $M_-^{( n, b )}$ for the $(1 + 1)$-dimensional interval with different mass parameters $m$ and a fixed discretization over the range $[ -4, 4 ]$ with resolution $n = 256$.}
\end{figure}
\autoref{fig:D2DoubleCone.Mm.mass} shows the diagonal of the smeared operator kernel for various values of $m$. It is clearly visible that the numerical results depend on $m$, and that they differ from the quadratic reference (the exact result suggested for all masses in \cite{2020LongoMorsella,2022Longo}).
All curves fall between the two reference lines, and they seem to approximate the quadratic reference (``massless case'') for small $m$ and the piecewise linear reference (``double wedge'') for large $m$.

Another perspective on the mass dependence is shown in \autoref{fig:D2DoubleCone.Mm.pos}, where the diagonal of $\innerProd{h_i}{M_-^{( n, b )} h_i}$ is plotted against the mass parameter for various positions $\mu_{i} = x$ within the interval region.
Note, however, that the data points for inverse masses of the order of the discretization resolution, and large inverse masses of the order of the discretization range should be considered less robust, because a detailed analysis of their behaviour would require larger discretization parameters $n$ and $b$, respectively.
\begin{figure}
	\centering
	\includegraphics{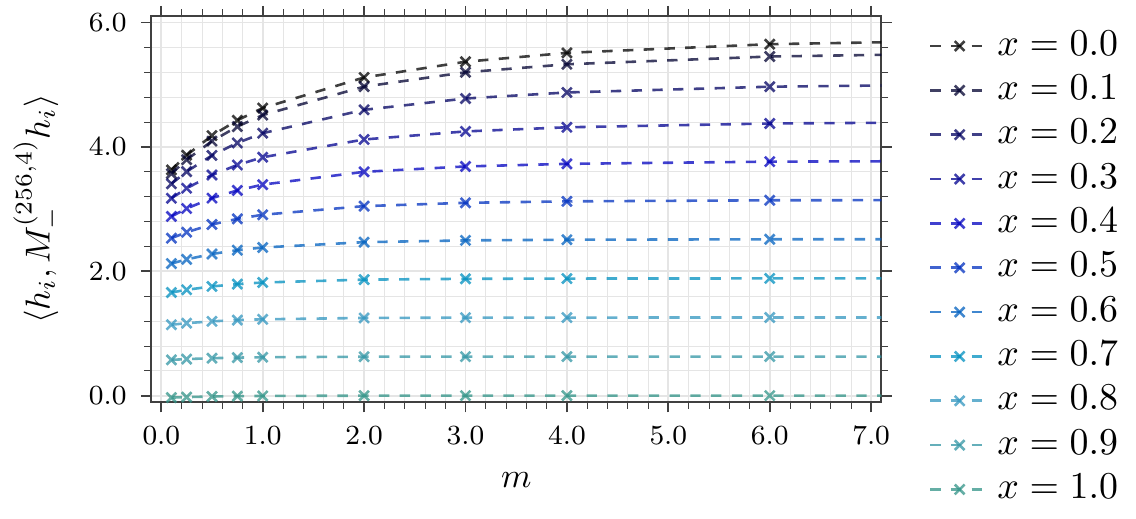}
	\caption{\label{fig:D2DoubleCone.Mm.pos} Comparison of the mass dependence of the smeared operator kernel $M_-^{( n, b )}$ at different positions $\mu_{i} = x$ within the right half of the interval in $1 + 1$ dimensions. The discretization is kept fixed over the range $[ -4, 4 ]$ with $n = 256$ box functions.}
\end{figure}

\section{The double cone in 3 + 1 dimensions}
\label{sec:D4DoubleCone}

We now investigate the analogous case of a double cone in physical spacetime dimension; that is, we set
\begin{align}
			\mathcal{H}_\mathrm{r}
	&= \mathrm{L}^2_\Reals (\Reals^3), \quad
&
			\mathcal{L}_\mathrm{r}
	&= \left\{ f \in \mathcal{H}_\mathrm{r} \bypred \supp f \subset \mathcal{B}_1 \right\},
&
			A
	&= -\Laplace + m^2
	 ,
\end{align}
where $\mathcal{B}_1$ is the ball of radius 1 around the origin.
However, since discretization with three-dimensional box functions is not numerically feasible with our current methods (the required matrix dimensions would be far too large), we first use rotational symmetry to simplify the problem.

To that end, let us express the Helmholtz operator $A$ in spherical coordinates,
\begin{subequations}
\begin{align}
\label{eq:ModHelmholtzD4}
			A
	&= - \frac{1}{r^2} \pderiv{}{r} r^2 \pderiv{}{r}
		+ \frac{L^2}{r^2}
		+ m^2
	,
\\\label{eq:AngularMomentumOperator}
			L^2
	&:= - \frac{1}{\sin \vartheta} \pderiv{}{\vartheta} \sin \vartheta \pderiv{}{\vartheta}
		- \frac{1}{\sin^2 \vartheta} \pderiv[2]{}{\varphi}
	.
\end{align}
\end{subequations}
Decomposing the square of the angular momentum operator $L^2$ into its known eigenbasis, i.e., spherical harmonics $\mathrm{Y}^{k}_{\ell}$ ($\ell \in \Naturals_0$, $k \in \Integers$, $-\ell \leq k \leq \ell$) with eigenvalues $\ell(\ell+1)$, we can identify by a unitary transformation,
\begin{align}
			\mathcal{H}_\mathrm{r}
	&= \bigoplus_{\ell, k} \mathrm{L}^2_\Reals\bigl( ( 0, \infty ), r^2 \d{r} \bigr)
	,
&
			\mathcal{L}_\mathrm{r}
	&= \bigoplus_{\ell, k} \left\{ f \bypred \supp f \subset ( 0, 1 ] \right\}
	,
&
			A
	&= \bigoplus_{\ell, k} A_{\ell}
	,
\end{align}
where $A_\ell$ is the modified spherical Bessel operator  
\begin{align}
\label{eq:ModHelmholtzD4.Radial}
			A_{\ell}
	&= - \frac{1}{r^2}
			\left(
				r^2 \pderiv[2]{}{r}
			+ 2 r \pderiv{}{r}
			- m^2 r^2
			- \ell (\ell + 1)
			\right)
	.
\end{align}
In this ``direct sum of one-particle structures'', it is clear that the conditions of \autoref{def:ops} hold for the sum if and only if they hold for every summand, and the operators $B$, $M_{\pm}$ etc.\ split accordingly.
For the following, we will focus on the summands for $\ell = 0$ and $\ell = 1$.

We proceed similarly as before to choose a discretization basis along the radial dimension, but now for the discretization range $[ 0, b ]$ with $b \in \{ 2, 4, 6 \}$.
Since the measure for the inner products of the Hilbert spaces takes the form $r^2 \d{r}$, the normalization factor of the box functions changes accordingly.
The discretization grid is defined in the same way as for the interval case in two dimensions -- equally spaced within the ball, $r \in [ 0, 1 ]$.
Again we present our results for $n = 256$ only, with $n = 128$ yielding coarser but very similar values. 
For $b = 2$, the outside spacing equals the inside, while for $b = 4$ and $b = 6$, the spacing increases linearly from $r = 1$ to $r = b$ as in \autoref{sec:D2DoubleCone}.
Thus $\dim \mathcal{L}_\mathrm{r} = \dim\mathcal{L}_\mathrm{r}^\perp$ holds summand by summand.

Instead of a discretization of $A_{\ell}^{-\isupfrac{1}{4}}$, we first consider the inverse $A_{\ell}^{-1}$, since the latter has a known expression for its kernel.
Namely, it is the Green's function that solves
\begin{align}
\label{eq:ModHelmholtzD4.Radial.GreensFunction.Equation}
			\big(A_{\ell} A_{\ell}^{-1}\big)( r, s )
	&= \frac{1}{r^2} \updelta( r - s )
	.
\end{align}
Expressed in terms of modified Bessel functions of the first and second kind, $\BesselI{\ell + \isubfrac{1}{2}}$ and $\BesselK{\ell + \isubfrac{1}{2}}$, respectively, we have
\begin{align}
\label{eq:ModHelmholtzD4.Radial.GreensFunction}
			A_{\ell}^{-1}( r, s )
	&= \sqrt{\frac{1}{r s}}
			\left(
				\upTheta( r - s )
				\BesselK{\ell + \subfrac{1}{2}}( m r )
				\BesselI{\ell + \subfrac{1}{2}}( m s )
			+ \upTheta( s - r )
				\BesselI{\ell + \subfrac{1}{2}}( m r )
				\BesselK{\ell + \subfrac{1}{2}}( m s )
			\right)
	.
\end{align}
The integrals for the matrix components $( A_{\ell}^{-1} )^{( n, b )}_{i j}$ are computed using standard anti-derivatives of the modified Bessel functions.
Notice that the operator kernel is a sum of products such that a coordinate transformation, as in \autoref{sec:D2RightWedge}, is not necessary.
Afterwards, the required discretization of $A_{\ell}^{-\isupfrac{1}{4}}$ is obtained by computing the fractional power of the matrix $( A_{\ell}^{-1} )^{( n, b )}_{i j}$ numerically.
Finally, $A_{\ell}^{\isupfrac{1}{4}}$ is computed as the matrix inverse of $( A_{\ell}^{-\isupfrac{1}{4}} )^{( n, b )}_{i j}$ to make sure that these two matrices are inverses of each other with a sufficient numerical precision.

Applying our algorithm, we obtain discretized approximations for $M_-^{( n, b )}$ as shown in \autoref{fig:D4DoubleCone.kernelMm} for $m = 1.0$ and $\ell \in \{ 0, 1\}$.
\begin{figure}
	\centering
	\begin{subfigure}{0.9\textwidth}
		\centering
		\includegraphics{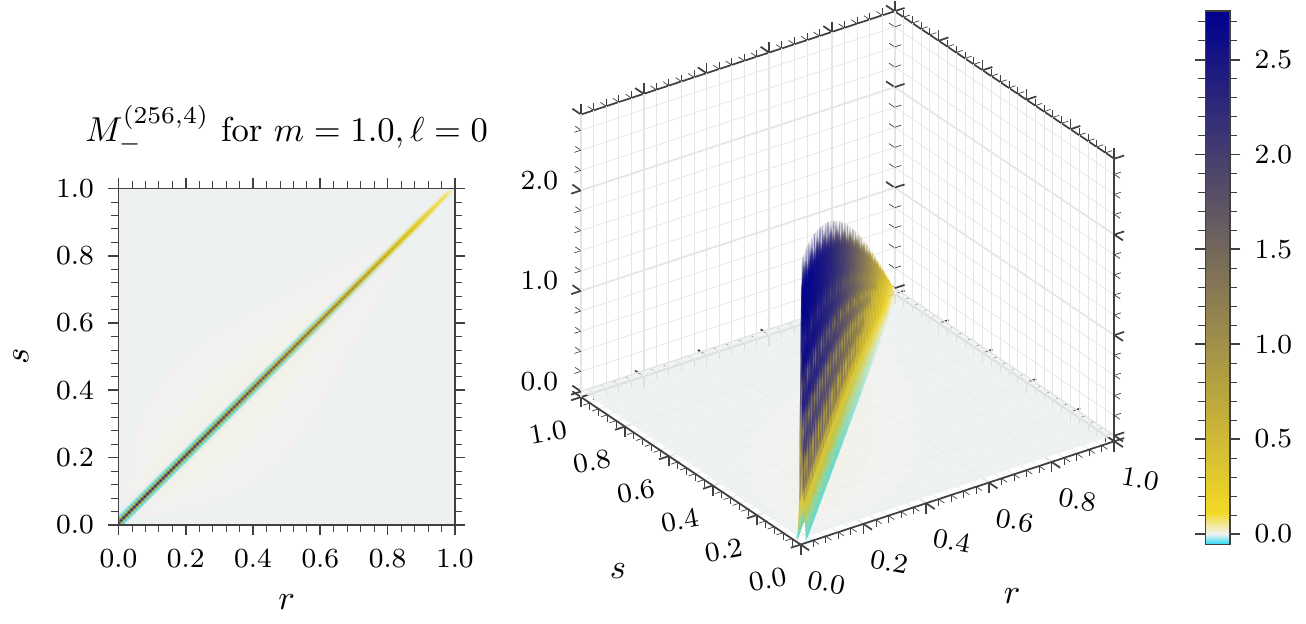}
		\caption{\label{fig:D4DoubleCone.kernelMm.ell0} $M_-^{( n, b )}$ for a discretization with $n = 256$ and $b = 4$ and for parameters $m = 1.0$ and $\ell = 0$.}
	\end{subfigure}
	\\[2em]
	\begin{subfigure}{0.9\textwidth}
		\centering
		\includegraphics{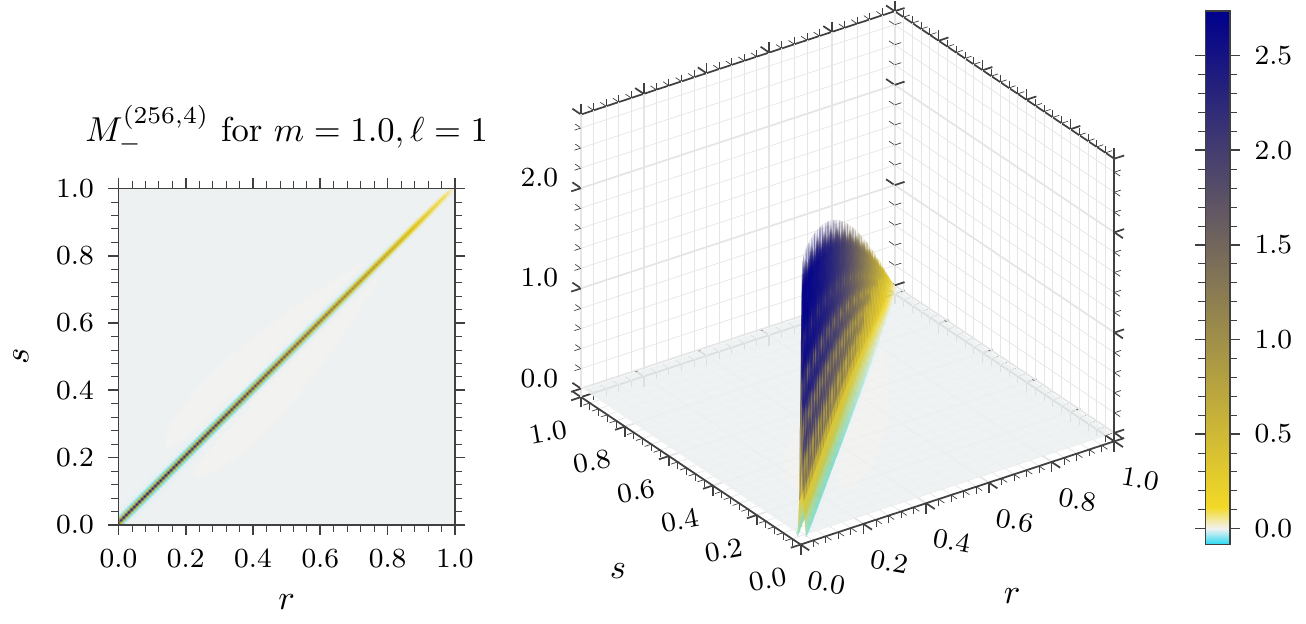}
		\caption{\label{fig:D4DoubleCone.kernelMm.ell1} $M_-^{( n, b )}$ for a discretization with $n = 256$ and $b = 4$ and for parameters $m = 1.0$ and $\ell = 1$.}
	\end{subfigure}
	\caption{\label{fig:D4DoubleCone.kernelMm} Example of the discretization of the radial operator kernel $M_-( r, s )$ for $\ell = 0$ (\ref{fig:D4DoubleCone.kernelMm.ell0}) and $\ell = 1$ (\ref{fig:D4DoubleCone.kernelMm.ell1}). The discretization uses 256 box functions over the range $[ 0, 4 ]$ with half of the functions supported on $r \in [ 0, 1 ]$ (equally spaced grid), and linearly increasing grid steps away from the interval (not shown in the plots). Similar to \autoref{fig:D2RightWedge.kernelMm} and \autoref{fig:D2DoubleCone.kernelMm}, the kernel is mostly diagonal (see matrix plot on the left) and falls off rapidly for $r \neq s$ (see surface plot on the right). Each pair of plots shares the same colour scale from cyan for small negative values, through very light grey at 0 and yellow for small positive values, to dark blue for positive values. Notice that the scales are almost (but not quite) identical for $\ell = 0$ and $\ell = 1$.}
\end{figure}
For both $\ell = 0$ and $\ell = 1$, these results look very similar to the two-dimensional interval, with a rapid falloff away from the diagonal $r = s$.

For a quantitative comparison along the diagonal, we swap the Gaussian test functions \eqref{eq:GaussianTestFunction} for log-Gaussian test functions (with parameters $\sigma$ and $\mu_{i}$ in natural scale),
\begin{align}
\label{eq:LogGaussianTestFunction}
			h_{i}( r )
	&= \frac{1}{\sqrt[4]{\pi \log \alpha_{i}}}
			\sqrt{\frac{1}{r^3}}
			\exp\left( -\frac{ \log^2\left( \alpha_{i} \frac{r}{\mu_{i}} \right)}{4 \log \alpha_{i}} \right)
	,
&
			\text{where}\quad
			\alpha_{i}
	&:= \sqrt{1 + \frac{\sigma^2}{\mu_{i}^2}}
	.
\end{align}
The position parameter $\mu_{i}$ takes values in $\{ 0.05, 0.1, 0.15, 0.2, \ldots, 1.15, 1.2 \}$, and we set $\sigma = \frac{6}{128}$.

\begin{figure}
	\centering
	\includegraphics{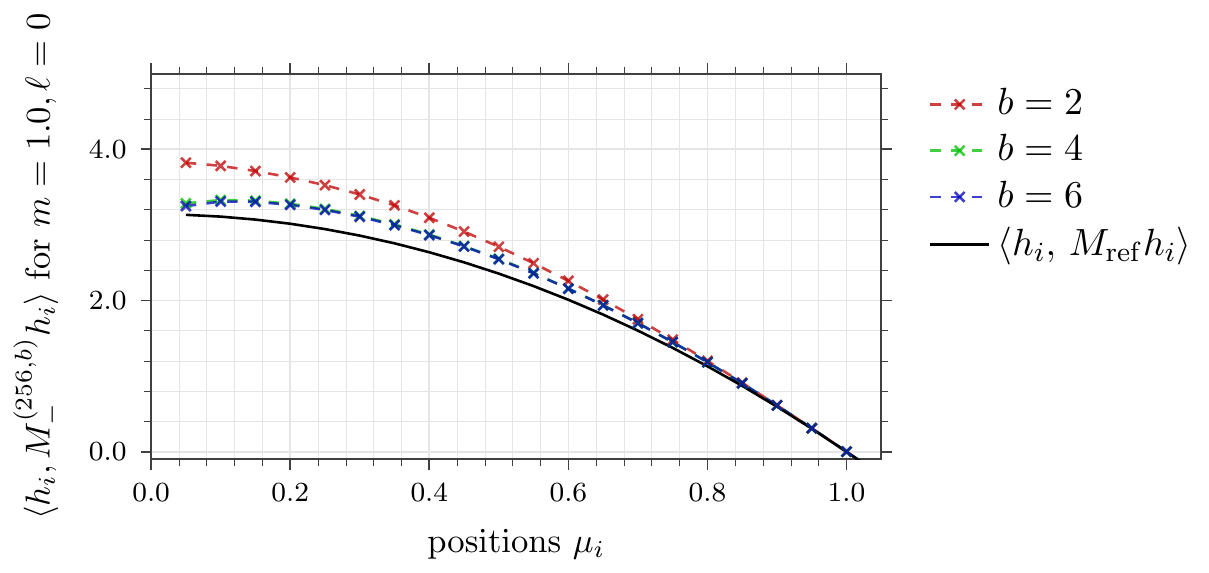}
	\caption{\label{fig:D4DoubleCone.Mm.disc} Comparison of different discretization ranges $b$ for the smeared expression $M_-^{( n, b )}$ (with $n = 256$, $m = 1.0$, and $\ell = 0$) in the $(3 + 1)$-dimensional case.}
\end{figure}
When varying the discretization parameter $b \in \{ 2, 4, 6 \}$ (while having the resolution parameter fixed to $n = 256$), see \autoref{fig:D4DoubleCone.Mm.disc}, we see that the results are nearly identical for cut-offs at 4 and 6, while there is a noticeable change when going from $b = 2$ to 4.
Also notice that we now have a discretization boundary at $r = 0$ that might cause slightly imprecise results for very small radii.
To take a closer look at the mass dependence, we choose $b = 4$.

As analytic reference, we take again the quadratic expression from the massless solution, where here it takes the form
\begin{align}
\label{eq:D4DoubleCone.ModularOperator.UpperRightBlock.Reference}
			M_{\mathrm{ref}}( r, s )
	&:= \pi \left( 1 - r^2 \right) \frac{1}{r^2} \updelta( r - s )
	.
\end{align}
This reference as well as the numerical results are smeared against the log-Gaussian test functions \eqref{eq:LogGaussianTestFunction}.

Though the mass dependence is not as pronounced as in the two-dimensional case, a variation with the parameter $m$ is clearly present and shown in \autoref{fig:D4DoubleCone.Mm.mass} for two different values $\ell = 0$ (\autoref{fig:D4DoubleCone.Mm.mass.ell0}) and $\ell = 1$ (\autoref{fig:D4DoubleCone.Mm.mass.ell1}).
\begin{figure}
	\centering
	\begin{subfigure}{0.9\textwidth}
		\centering
		\includegraphics{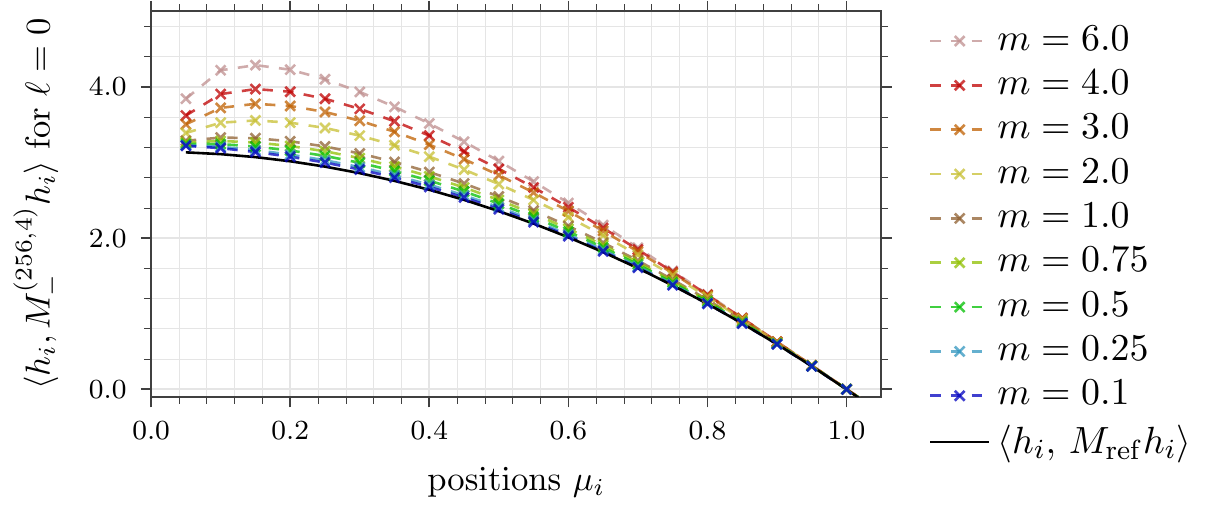}
		\caption{\label{fig:D4DoubleCone.Mm.mass.ell0} $M_-^{( n, b )}$ for a discretization with $n = 256$ and $b = 4$ and for parameter $\ell = 0$.}
	\end{subfigure}
	\\[2em]
	\begin{subfigure}{0.9\textwidth}
		\centering
		\includegraphics{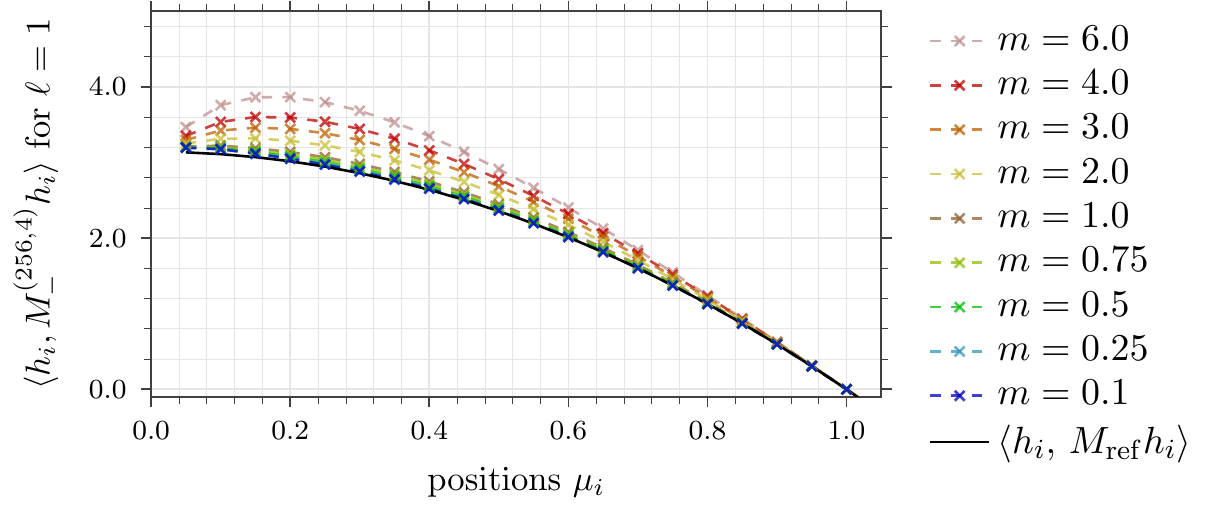}
		\caption{\label{fig:D4DoubleCone.Mm.mass.ell1} $M_-^{( n, b )}$ for a discretization with $n = 256$ and $b = 4$ and for parameter $\ell = 1$.}
	\end{subfigure}
	\caption{\label{fig:D4DoubleCone.Mm.mass} Comparison of the smeared operator kernel $M_-^{( n, b )}$ for different mass parameters $m$ and a fixed discretization of the range $[ 0, 4 ]$ with $n = 256$ box functions, half of them inside the interval.}
\end{figure}
Once again, the curves seem to approach the massless reference when $m$ becomes small, but diverge from it when the mass is increased.
In \autoref{fig:D4DoubleCone.Mm.pos}, we present the same dependence from the perspective of various masses at fixed locations $\mu_{i} = r$ within the interval.
Especially for small radii, a mass dependence emerges, while the nearly horizontal lines for $r > 0.7$ shows that the modular operator becomes mass-independent when approaching the interval boundary $r = 1$.
At that boundary, the behaviour is very similar to the case of a left wedge placed at $r = 1$.
\begin{figure}
	\centering
	\begin{subfigure}{0.9\textwidth}
		\centering
		\includegraphics{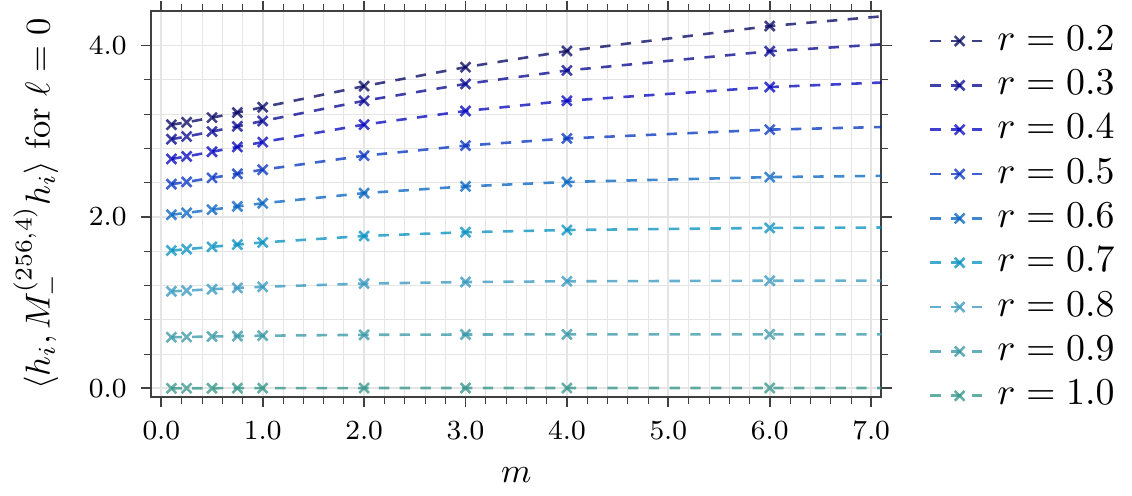}
		\caption{\label{fig:D4DoubleCone.Mm.pos.ell0} $M_-^{( n, b )}$ for a discretization with $n = 256$ and $b = 4$ and for parameter $\ell = 0$.}
	\end{subfigure}
	\\[2em]
	\begin{subfigure}{0.9\textwidth}
		\centering
		\includegraphics{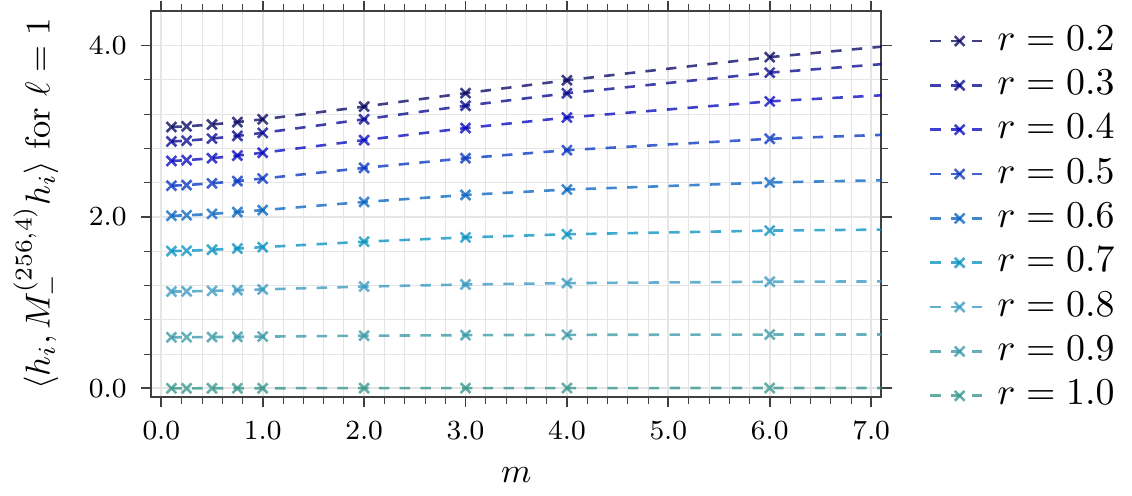}
		\caption{\label{fig:D4DoubleCone.Mm.pos.ell1} $M_-^{( n, b )}$ for a discretization with $n = 256$ and $b = 4$ and for parameter $\ell = 1$.}
	\end{subfigure}
	\caption{\label{fig:D4DoubleCone.Mm.pos} Different perspective on the mass dependence as shown in \autoref{fig:D4DoubleCone.Mm.mass}, here at different positions $\mu_{i} = r$.}
\end{figure}
Comparing the plots for $\ell = 0$ and $\ell = 1$ at equal masses, notice that there is also a minor dependence on the angular eigenvalue $\ell$, though this effect is less distinct.

\section{Conclusions}
\label{sec:Conclusion}

In this article, we have computed a numerical approximation to the modular generator for a double cone in the massive free field on $( 1 + 1 )$- and $( 3 + 1 )$-dimensional Minkowski space. Specifically, we approximated the component $M_-$ as in \eqref{eq:modGenFormula}. Using a finite-dimensional approximation of the one-particle structure, we discretized the kernel of $M_-$ in a basis of box functions in position space, reducing the problem to functional calculus of matrices.   

Our results in the $(1+1)$-dimensional case indicate that the leading contribution of $M_-$ is a multiplication operator, at least on the subspace corresponding to the interior of the double cone. 
However, if so, this multiplication operator must be mass-dependent, unlike expected in \cite{2020LongoMorsella,2022Longo}. Similarly in $3+1$ dimensions, $M_-$ seems to act at leading order as a mass-dependent multiplication operator on every subspace of fixed angular momentum $\ell$; 
if, however, the multiplier function also depends on $\ell$, as our results indicate, then $M_-$ is not a multiplication operator in the usual sense. 
Irrespective of that, the difference of the multiplication parts for different masses would have continuous spectrum, making it unlikely that compact perturbation methods \cite{2022Longo} are applicable.

Apart from the mass dependence of the multiplication or ``diagonal'' part, one may ask whether the modular Hamiltonian of the double cone has an additional, non-diagonal contribution to its kernel, stemming perhaps from a non-geometric action of the modular group. We do not claim to resolve this question, but note that the non-diagonal values appear to be at least some orders of magnitude smaller than the diagonal ones, as might perhaps be expected. With our present resolution, we  are not able to confirm whether or not they are zero, leaving this as an open problem for ongoing research.

We also mention an application to relativistic quantum information theory:  In our situation, the relative entropy between the Fock vacuum $\omega$ and a coherent excitation $\omega_f(\cdot) :=\omega(W(f)^\ast \cdot W(f))$, $f\in \mathcal{H}$, with respect to the Weyl subalgebra $\mathfrak{W}(\mathcal{L})\subset \mathfrak{W}(\mathcal{H})$, is given by \cite{2020CiolliLongoRuzzi} 
\begin{subequations}
\begin{align}
			S_{\mathfrak{W}(\mathcal{L})}( \omega_f \mid\mid \omega )
	&= -\innerProd[\mathcal{H}]{f}{P^\ast \log \Delta f}
\\
	&= \innerProd[\subfrac{1}{4} \oplus -\subfrac{1}{4}]{%
				f
			}{%
				\left(
					\chi_{-\subfrac{1}{4}} M_+
				\oplus \chi_{\subfrac{1}{4}} M_-
				\right) f
			}
	;
\end{align}
\end{subequations}
thus our methods provide a numerical approximation of this quantity. 

Similar numerical studies of the entanglement Hamiltonian and entanglement entropy \cite{2020EislerDigiulioTonniPeschel,2022JaverzatTonni} start from a lattice approach to the quantum system, while we set out from the quantum field theory in the continuum. 
Besides this conceptual difference, we note that those authors obtain the entanglement Hamiltonian as an operator on the subspace $\mathcal{L}$, while our technique yields the modular operator for $\mathcal{L}$ acting on the full Hilbert space $\mathcal{H}$. 
An in-depth comparison of these two numerical approaches would be worthwhile, but requires further work.

While we verified the quality of our approximation in the test case of a $( 1 + 1 )$-dimensional wedge, where the analytical result is well known, we did not present any rigorous convergence proof. 
Certainly, it would be of interest to establish precise conditions under which the modular objects of a finite-dimensional one-particle structure converge to those of an infinite-dimensional one; we hope to return to this point elsewhere.

We restricted our attention to double cones in the scalar field, but the same numerical methods should apply to other space-time regions, other wave operators, and by purification methods (cf.~\cite{2022BostelmannCadamuroDelVecchio}) also to thermal states and to quasifree states of linear quantum fields in a gravitational background, yielding relevant results for relative entropies in these cases. 
However, in particular for non-symmetric regions in higher dimensions, this would require a performance optimization of the numerical algorithm, which we have not focussed on here. 
In particular, it would be worthwhile to find an alternative approach to the eigenvalue problem of the matrices $B^{(n)}$ that eliminates the need for high-precision floating point arithmetic.

\section*{Acknowledgements}
We would like to thank Roberto Longo and Gerardo Morsella for discussing their work on the modular operator for double cones with us, 
Ko Sanders for discussions on the modular data of one-particle structures, 
Detlev Buchholz for comments on preliminary results, 
and Robin Hillier for helpful suggestions regarding the presentation. 
D.C.\ and C.M.\ are supported by the Deutsche Forschungsgemeinschaft (DFG) within the Emmy Noether grant CA1850/1-1. H.B.~would like to thank the Institute for Theoretical Physics at the University of Leipzig for hospitality.

\printbibliography

\end{document}